\documentclass[11pt,a4paper]{article}

\usepackage{mies}

\usepackage{amsthm}
\usepackage{enumitem}

\newtheorem{theorem}{Theorem}
\newtheorem{proposition}{Proposition}

\newtheorem{corollary}{Corollary}
\theoremstyle{definition}
\newtheorem{definition}{Definition}
\newtheorem{assumption}{Assumption}
\theoremstyle{remark}

\title{Vibe Coding Kills Open Source\thanks{We thank Gergely Orosz, Eduardo Arino de la Rubia and Kevin Xu for helpful comments. This project was funded by the European Research Council (ERC Advanced Grant agreement number 101097789). The views expressed in this research are those of the authors and do not necessarily reflect the official view of the European Union or the European Research Council. Project no.\ 144193 has been implemented with the support provided by the Ministry of Culture and Innovation of Hungary from the National Research, Development and Innovation Fund, financed under the KKP\_22 funding scheme.  Békés thanks NRDIF for financial support under OTKA Project \#147121 .}}

\author{
Miklós Koren\thanks{Central European University, KRTK, CEPR and CESifo. \emph{Corresponding author:} miklos@koren.work}
\and Gábor Békés\thanks{Central European University, KRTK, and CEPR.}
\and Julian Hinz\thanks{Bielefeld University and Kiel Institute for the World Economy.}
\and Aaron Lohmann\thanks{Kiel Institute for the World Economy.}
}

\date{\today}

\hypersetup{
    pdfauthor = {Miklós Koren, Gábor Békés, Julian Hinz, Aaron Lohmann},
    pdftitle  = {Vibe Coding Kills Open Source},
}

\begin{document}

\maketitle

\vspace{2em}

\begin{abstract}
Generative AI is changing how software is produced and used. In vibe coding, an AI agent builds software by selecting and assembling open-source software (OSS), often without users directly reading documentation, reporting bugs, or otherwise engaging with maintainers. We study the equilibrium effects of vibe coding on the OSS ecosystem. We develop a model with endogenous entry and heterogeneous project quality in which OSS is a scalable input into producing more software. Users choose whether to use OSS directly or through vibe coding. Vibe coding raises productivity by lowering the cost of using and building on existing code, but it also weakens the user engagement through which many maintainers earn returns. When OSS is monetized only through direct user engagement, greater adoption of vibe coding lowers entry and sharing, reduces the availability and quality of OSS, and reduces welfare despite higher productivity. Sustaining OSS at its current scale under widespread vibe coding requires major changes in how maintainers are paid.
\end{abstract}

\vfill



\newpage

\section{Introduction}
\label{introduction}

Generative AI is reshaping software development. AI coding assistants such as Claude Code, Cursor, and Lovable let users translate intent into working applications with little or no manual coding. This AI-mediated mode of building software is often called ``vibe coding.'' Vibe coding reduces the cost of producing software, but it also changes how users interact with the software ecosystem. Traditionally, a developer selects packages, reads documentation, and interacts with maintainers and other users. Under vibe coding, an AI agent can select, compose, and modify packages end-to-end, and the human developer may not know which upstream components were used.

This shift raises a general equilibrium question about the sustainability of open source software (OSS). OSS is a nonrival input into producing more software, and it generates large social value relative to its direct production cost \citep{hoffmann-nagle-zhou-2024}. Yet many OSS projects rely on visibility and engagement from direct users---documentation visits, bug reports, public Q\&A, and reputation---to sustain maintenance and capture private returns \citep{lerner-tirole-2002,jones-2020-eseur}. If AI mediation substitutes for direct interaction, then the technology that makes software easier to use may simultaneously erode the engagement-based channel that funds and motivates its supply. We ask whether the productivity gains from vibe coding outweigh the loss of appropriable demand once developer entry and selection respond.

We build a tractable model of the OSS ecosystem in the spirit of monopolistic competition with endogenous variety \citep{krugman-1980} and heterogeneous-project selection \citep{melitz-2003}. Developers incur an up-front development cost before project quality is realized, draw project quality from a Pareto distribution, and then decide whether to release the project as OSS by paying a fixed sharing cost. Users value variety and choose among available packages under a discrete-choice demand system. In addition, after selecting a package, users choose whether to interact with it directly or through an AI agent. This nested choice captures that vibe coding is a usage technology that affects both user utility and the degree of engagement that maintainers can monetize.

Vibe coding affects the OSS ecosystem through two channels. The first is a productivity channel: by reducing the effective cost of using a given package, AI raises user utility and lowers the cost of producing new software that builds on existing components. Even before state-of-the-art vibe coding tools, field experiments documented sizable productivity gains from AI coding assistance \citep{peng-2023-copilot-productivity,cui-demirer-jaffe-musolff-peng-salz-2025}. The second is a demand-diversion channel: when users rely on AI agents rather than direct interaction, maintainers capture less engagement and therefore less private return per unit of usage. These channels interact because entry and sharing are endogenous. Better usage technology raises the value of the ecosystem as an intermediate input and would tend to increase entry. But a contraction in monetizable engagement shrinks the effective market for OSS maintainers and discourages sharing and entry. Since ``customer demand motivates the supply of energy that drives software ecosystems'' \citep{jones-2020-eseur}, a technology that alters how demand is expressed can reshape supply even if total usage rises.

Our main result is that under traditional OSS business models, where maintainers primarily monetize direct user engagement (higher visibility leading to paid opportunities or other forms of appreciation), higher adoption of vibe coding reduces OSS provision and lowers welfare. In the long-run equilibrium, mediated usage erodes the revenue base that sustains OSS, raises the quality threshold for sharing, and reduces the mass of shared packages. Variety shrinks and the average quality of shared OSS falls, so user utility can decline despite better AI. The decline can be rapid because the same magnification mechanism that amplifies positive shocks to software demand also amplifies negative shocks to monetizable engagement. In other words, feedback loops that once accelerated growth now accelerate contraction.

The key intuition behind the negative channel is that when mediated adoption is responsive to AI capability, direct engagement can fall faster than development costs decline. Utility falls because entry decisions generate a business-stealing externality: when a new package enters, it attracts users away from existing packages, but the entrant does not internalize that this diversion reduces the engagement-based returns that other maintainers rely on. The strength of this diversion depends on how responsive users are to improvements in AI capability: when AI tools get better, how quickly do users switch away from direct interaction? The rapid diffusion of AI coding assistants suggests high responsiveness; in the policy section, we discuss how to quantify it using adoption and usage data \citep{demirer-fradkin-tadelis-peng-2025}.

We also analyze two extensions. First, we consider a benchmark in which vibe coding is used only by professional software developers and does not directly mediate final-user consumption. In this case, the demand-diversion channel is absent: AI lowers development costs without eroding engagement-based monetization. The equilibrium features higher entry and higher average quality of shared OSS. This benchmark helps interpret short-run dynamics in which AI is primarily a developer productivity tool, while vendors increasingly target non-developer users as well \citep{anthropic-2026-cowork}. Second, we allow for alternative monetization arrangements and ask what is the lowest per-user monetization that sustain OSS entry at its current level. This can be lower than the current level, because the increase in productivity partly offsets the loss of developer reward. But, under plausible parameters, per-user monetization must remain very close to current levels to avoid large declines in OSS provision.

To keep the analysis tied to how OSS is produced and used, we tailor the model to four salient industry features: large fixed costs and nonrivalry, user love of variety, heavy-tailed heterogeneity in project outcomes with endogenous selection into sharing, and software as an intermediate input into producing more software (a ``software-begets-software'' feedback). These features are quantitatively important for OSS, where most repositories attract no attention and a small upper tail accounts for a large share of usage. 

The remainder of this paper is structured as follows. Section~\ref{sec:industry} documents industry trends and institutional details that motivate the model, including rapid AI diffusion, engagement substitution under AI mediation, and concentration in OSS outcomes. Section~\ref{a-model-of-the-open-source-software-ecosystem} develops the model and derives equilibrium and welfare results in the baseline economy and under vibe coding. Section~\ref{calibration} discusses parameter choices and illustrates the quantitative implications. Section~\ref{sec:conclusion} concludes and discusses implications for platform design and policy.
\section{Industry Trends and Institutional Details}\label{sec:industry}

This section provides institutional background and descriptive patterns that discipline the model. We emphasize three facts. First, AI-assisted coding has diffused rapidly and its capabilities have improved quickly, making the technology shock quantitatively large. Second, AI mediation appears to substitute for direct engagement with upstream knowledge producers, weakening the engagement-based channel through which many open source software (OSS) projects capture private returns. Third, outcomes in OSS are highly concentrated, consistent with selection into sharing and heavy-tailed project ``quality''.

\subsection{AI Coding Adoption and Capability Growth}

AI-assisted coding has become embedded in production at large software organizations. By October 2024, more than a quarter of all new code at Google was generated by AI and subsequently reviewed and accepted by engineers \citep{pichai-2024}. At Anthropic, CEO Dario Amodei reported even higher shares in September 2025, stating that ``70, 80, 90 percent of the code is written by Claude'' \citep{amodei-2025}. Consistent with these adoption rates reflecting a productivity shifter rather than full automation, qualitative evidence indicates that experienced developers use agents as assistance tools while retaining control over system design and core implementation decisions \citep{huang-2025-control}.

Micro-level evidence from GitHub corroborates rapid diffusion. Using commit-level data from 170{,}000 developers, \citet{daniotti-2025-ai-code} estimate that by the end of 2024, AI generated roughly 29--30\% of Python functions authored by U.S.\ contributors, and adoption is associated with a 3.6\% increase in quarterly code contributions.

Capabilities have improved quickly along relevant dimensions. On SWE-bench,\footnote{\href{https://www.swebench.com}{https://www.swebench.com}} which evaluates issue resolution on real GitHub repositories, the best-performing model in 2024 (Claude~2) solved only 1.96\% of issues \citep{jimenez-2024-swebench}. By early 2026, the same model family (Claude~4.5) solved 74.2\% of issues. Rapid improvements in model capability help explain why adoption can accelerate sharply once performance crosses a usability threshold.

Figure~\ref{fig:ai-adoption} summarizes adoption patterns using two complementary sources. Panel (a) reports the share of businesses using AI by industry from the U.S.\ Census Bureau’s Business Trends and Outlook Survey (BTOS), showing that information-intensive industries lead adoption. Panel (b) reports AI intensity in software development based on GitHub commit data from \citet{daniotti-2025-ai-code}. In the model, these patterns motivate treating AI as a large, fast shock to the productivity of using and composing existing software.

\begin{figure}[t]
\centering
\begin{minipage}[t]{0.48\textwidth}
\centering
\caption*{(a) Business AI adoption by industry}
\vspace{-0.3em}
\includegraphics[width=\textwidth]{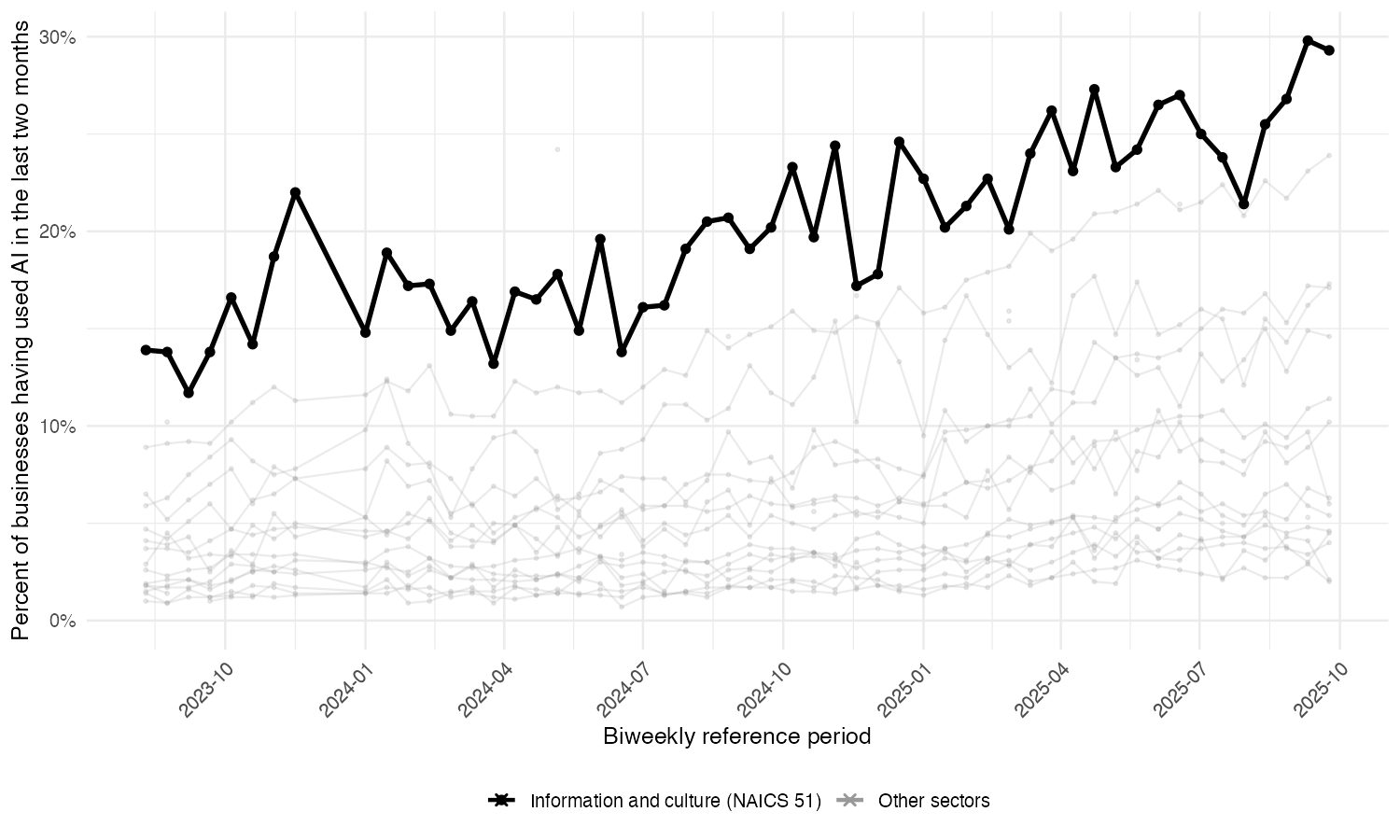}
\end{minipage}\hfill
\begin{minipage}[t]{0.48\textwidth}
\centering
\caption*{(b) AI use in software development}
\vspace{-0.3em}
\includegraphics[width=\textwidth]{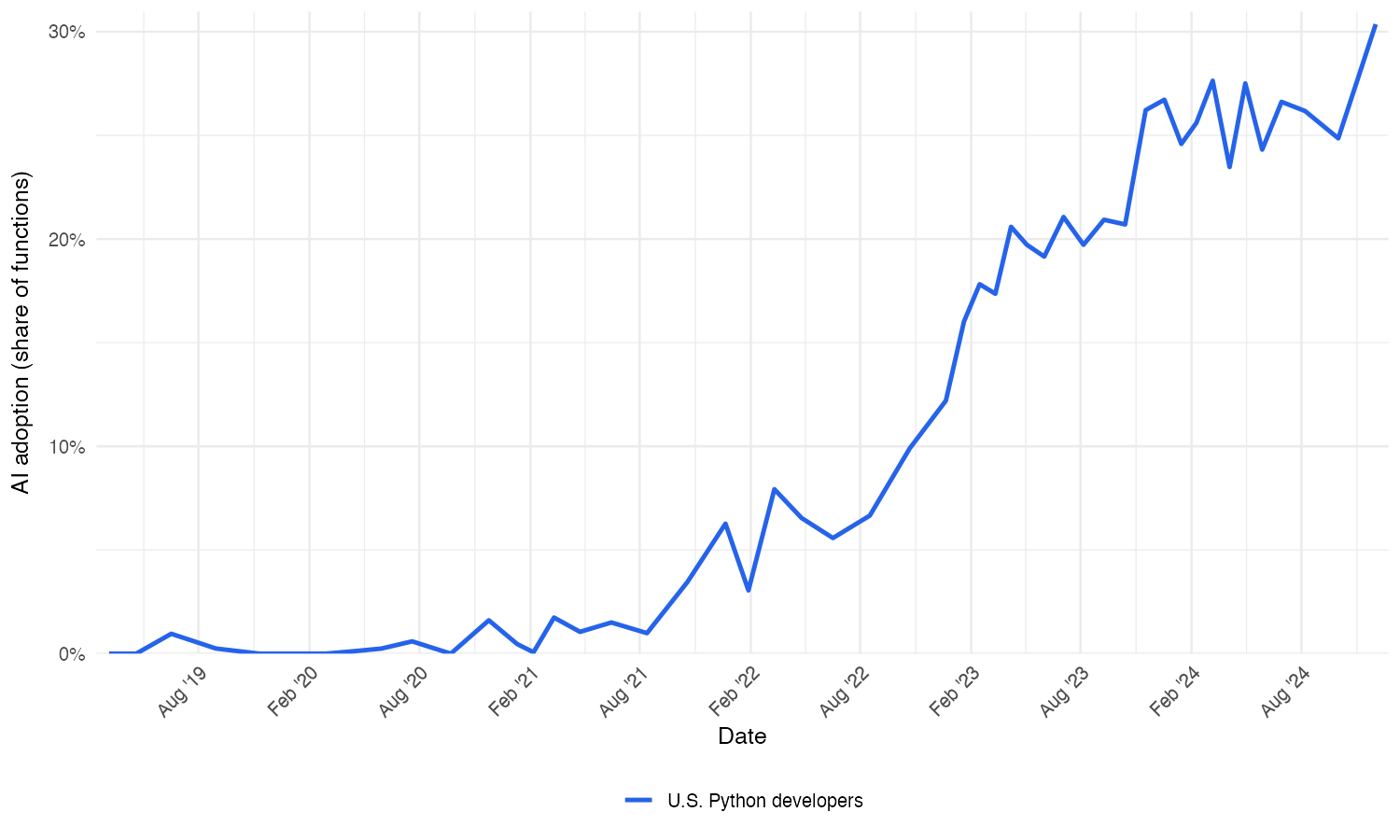}
\end{minipage}

\caption{Adoption of AI across industries and within software development.}
\label{fig:ai-adoption}

\begin{flushleft}
\footnotesize
\emph{Notes:} Panel (a) reports the share of businesses using AI by industry, based on survey evidence from the U.S.\ Census Bureau’s Business Trends and Outlook Survey (BTOS); information-intensive industries exhibit the highest adoption rates. Panel (b) reports the estimated share of software code written with AI assistance using GitHub commit data from \citet{daniotti-2025-ai-code}.
\end{flushleft}
\end{figure}

\subsection{How OSS Maintainers Capture Private Returns}

OSS generates large social value relative to its direct production cost \citep{hoffmann-nagle-zhou-2024}. Individual maintainers and small teams therefore typically sustain projects by capturing only a small fraction of the surplus they create, through channels that are closely tied to user visibility and engagement \citep{lerner-tirole-2002,jones-2020-eseur}. These channels include reputation and career opportunities (downloads, stars, citations), consulting and support leads that originate from documentation and community interaction, and paid complements such as hosted versions or enterprise add-ons. Many of these channels rely on direct interaction: a user who reads documentation, opens an issue, or asks a public question creates attention and signals demand.

The importance of engagement-based monetization varies across market segments. A large share of monetized activity around OSS occurs in enterprise services: the large enterprises segment accounted for more than 69.0 percent of the 2022 open source services market \citep{grand-view-research-2022-oss-services}. In contrast, a large fraction of OSS usage occurs in small-team and individual production, where adoption may be high but monetization relies on discovery through community channels rather than formal procurement.

In the model, we summarize private returns by a reward per direct user, and we allow AI mediation to reduce the extent to which usage translates into monetizable engagement.

\subsection{Engagement Substitution Under AI Mediation}

The model’s long-run mechanism relies on a wedge between OSS usage and the direct engagement that sustains maintainers. If AI-mediated assistance substitutes for direct interaction, the monetizable engagement margin can shrink even as adoption expands.

For public Q\&A platforms, \citet{del-rio-chanona-2024} provide causal evidence that access to ChatGPT reduced Stack Overflow activity by about 25 percent within six months relative to counterfactual platforms with limited access. Since then, overall Stack Overflow traffic has continued to decline sharply \citep{Orosz2025}.

Figure~\ref{fig:engagement-substitution} illustrates a similar decoupling for Tailwind CSS. Weekly npm downloads rise steadily, indicating growing usage of the framework, while public questions tagged with \texttt{tailwind-css} decline. Tailwind’s creator reports that ``traffic to our docs is down about 40\% from early 2023 despite Tailwind being more popular than ever'' and that revenue is ``down close to 80\%'' \citep{wathan-2026}. Together, these patterns suggest that AI mediation can divert interaction away from the surfaces where OSS projects monetize and recruit contributors.

\begin{figure}[t]
\centering
\begin{minipage}[t]{0.48\textwidth}
\centering
\caption*{(a) Tailwind: usage vs public Q\&A}
\vspace{-0.3em}
\includegraphics[width=\textwidth]{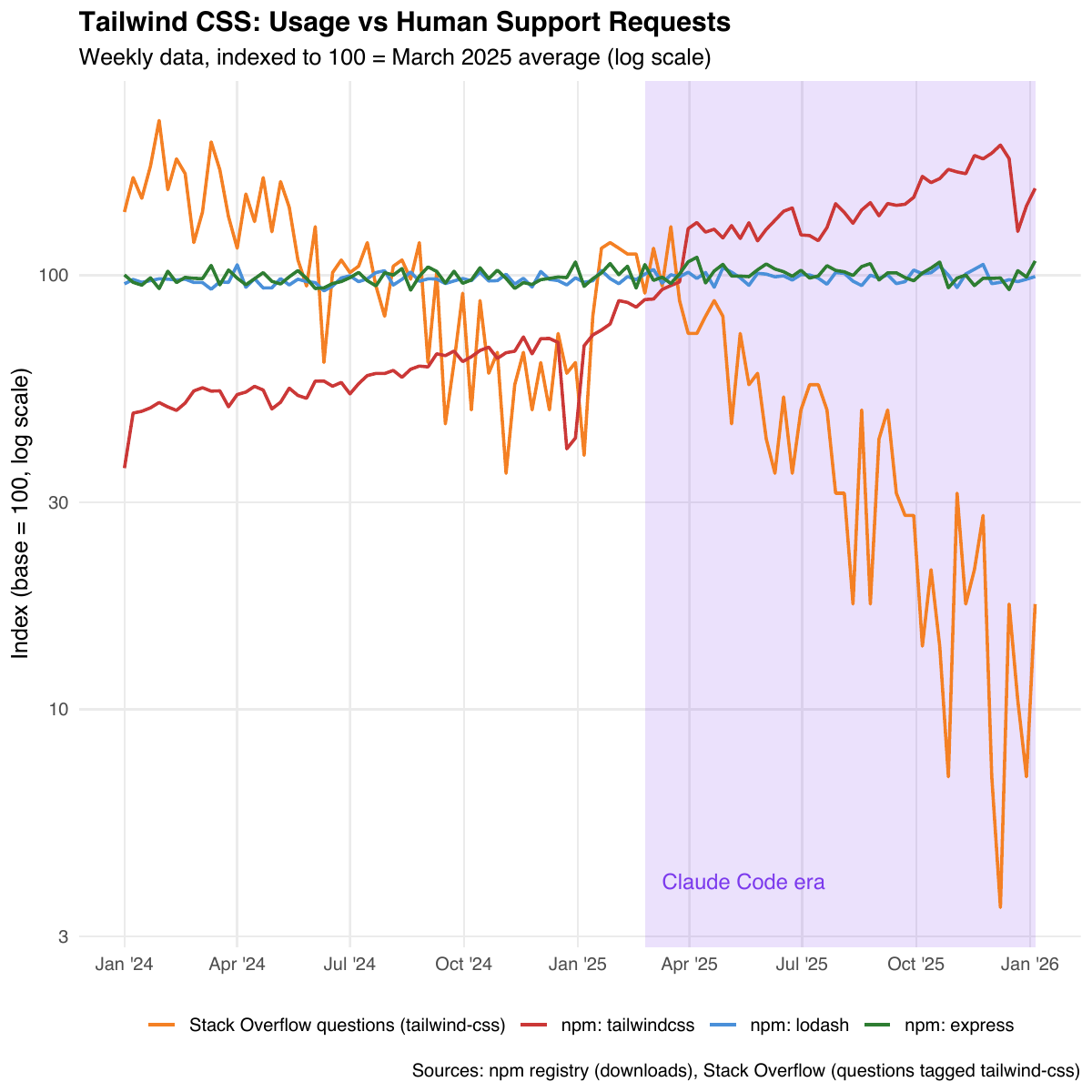}
\end{minipage}\hfill
\begin{minipage}[t]{0.48\textwidth}
\centering
\caption*{(b) Stack Overflow questions over time}
\vspace{-0.3em}
\includegraphics[width=\textwidth]{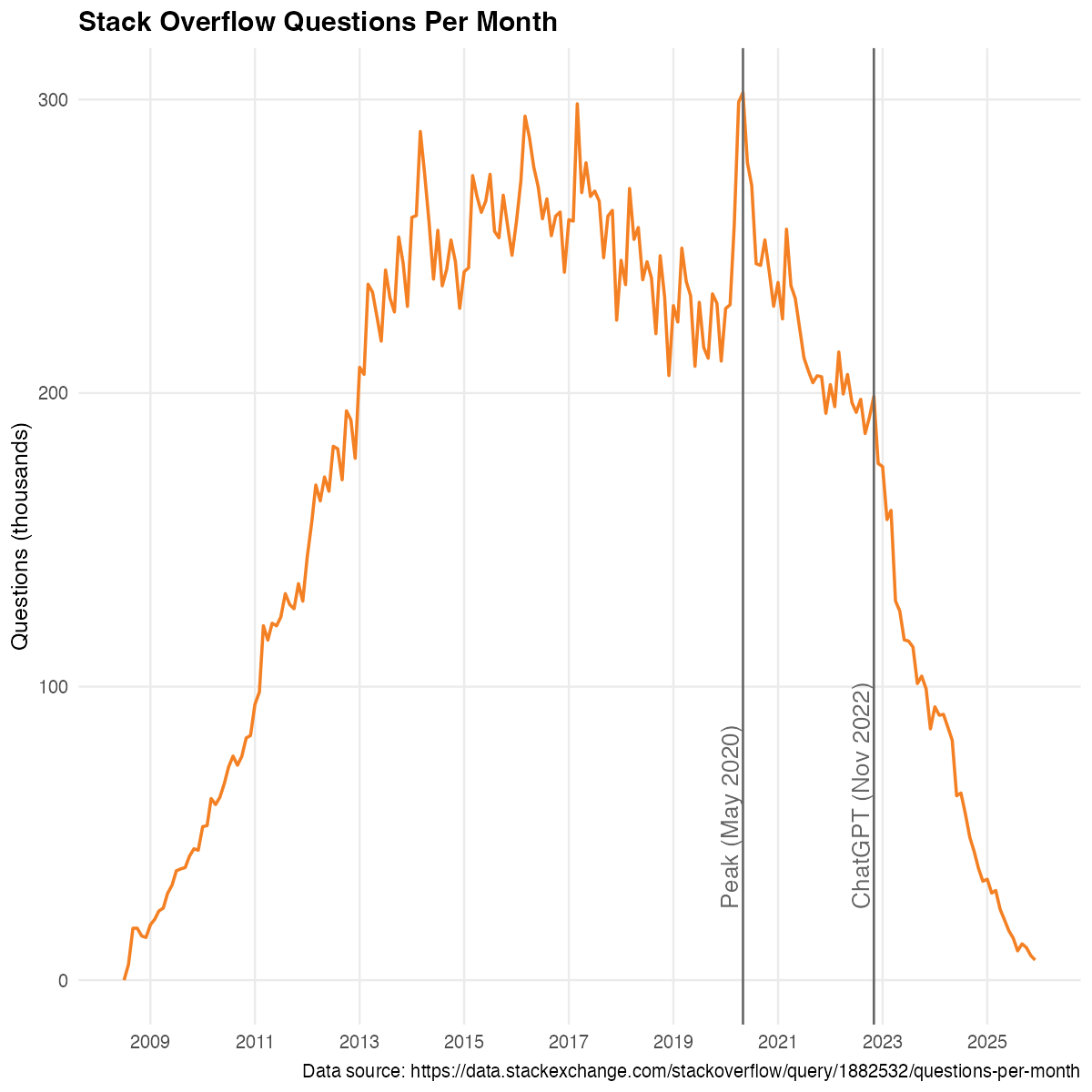}
\end{minipage}

\caption{Usage grows while engagement falls.}
\label{fig:engagement-substitution}
\begin{flushleft}
\footnotesize
\emph{Notes:} Panel (a) compares Tailwind CSS usage (npm downloads) to a measure of public engagement (Stack Overflow questions tagged with \texttt{tailwind-css}). Panel (b) shows the broader decline in Stack Overflow question volume. \citet{del-rio-chanona-2024} provide causal evidence of a post-ChatGPT decline in activity; \citet{wathan-2026} reports a concurrent decline in Tailwind documentation traffic and revenue.
\end{flushleft}
\end{figure}

\subsection{Concentration and Selection in OSS Outcomes}

Outcomes in open source software are highly unequal. Figure~\ref{fig:github-pareto}
plots the distribution of GitHub repository success using stars and downstream
dependencies as proxies for attention and adoption. The figure exhibits an
approximately linear relationship on log--log scales, indicating Pareto-like
behavior in the upper tail: a small fraction of projects attracts a
disproportionate share of attention, usage, and downstream reliance.

\begin{figure}[t]
\centering
\includegraphics[width=0.85\textwidth]{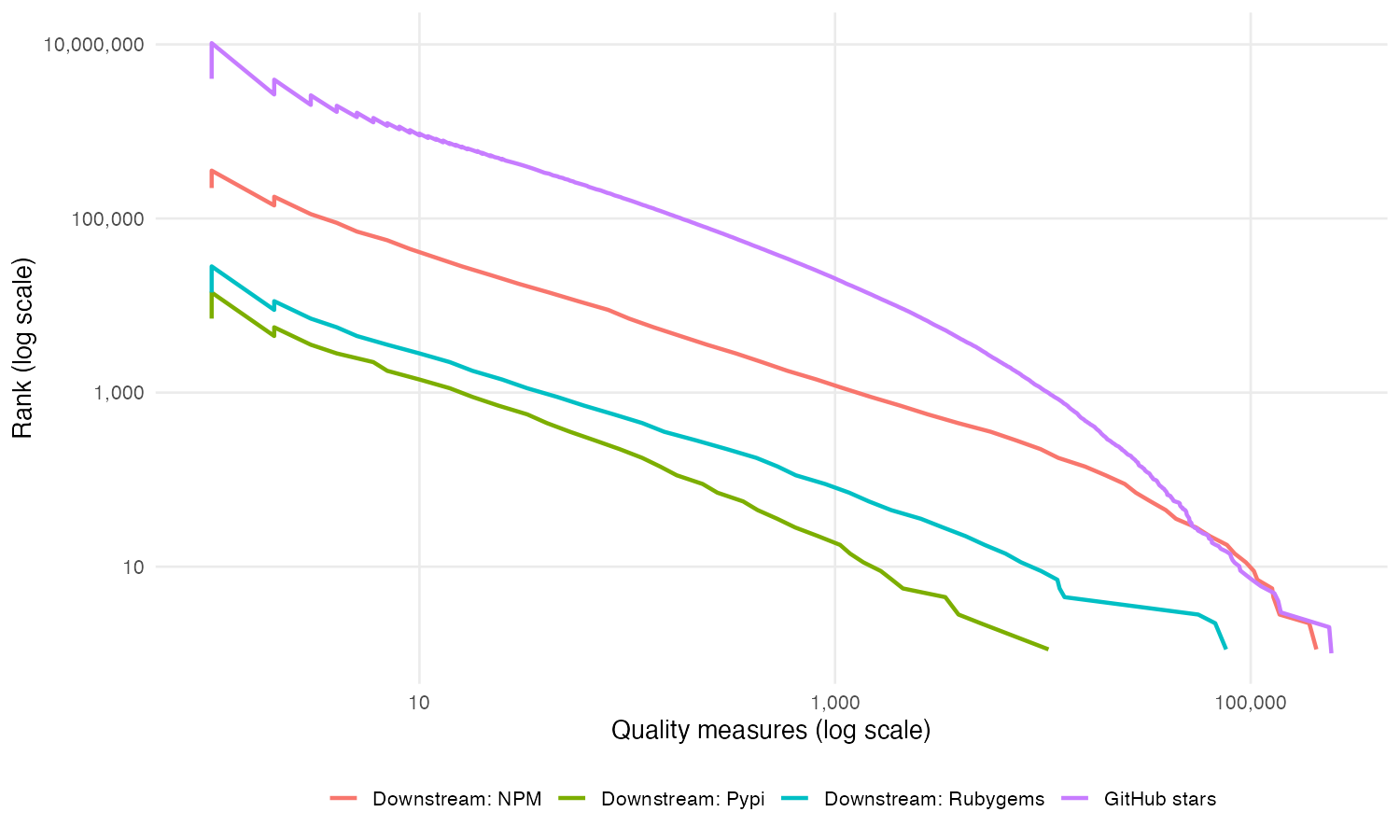}
\caption{Concentration of attention and usage in open source software.}
\label{fig:github-pareto}
\begin{flushleft}
\footnotesize
\emph{Notes:} The figure plots log--log rank--size relationships for GitHub
repositories using stars (attention) and downstream dependencies (usage) as
proxies for project success. Repositories are ranked by outcome, and the figure
reports median values within logarithmic rank bins. The approximately linear
relationship on log--log scales indicates Pareto-like behavior in the upper tail,
with extreme heterogeneity in outcomes. Data sources: GHTorrent \citep{gousios-2013-ghtorrent} and
Libraries.io \citep{libraries-io-2025}.
\end{flushleft}
\end{figure}

Outcomes are concentrated not only in the upper tail but also on the extensive margin: most repositories attract no attention and are never reused. Of approximately 125 million GitHub repositories, only 10.32 million have ever received at least one star. Restricting attention to JavaScript projects, GHTorrent records 13.4 million repositories, of which only 1.28 million have at least one star \citep{gousios-2013-ghtorrent}. Using Libraries.io \citep{libraries-io-2025}, we further observe that 840{,}419 JavaScript repositories are published as npm packages.

Downstream dependence is even more concentrated. Only 392{,}828 projects ever appear as a dependency of another project. Thus, only a small fraction of repositories are both shared and reused, consistent with large fixed costs of packaging and documentation and with strong selection into economically relevant OSS. In the model, these facts motivate a quality distribution with a heavy right tail and an endogenous sharing cutoff.

In addition to the graphical rank–size evidence, Table~\ref{tab:stars-tail-reg} quantifies the heaviness of the upper tail of OSS attention using GitHub stars. We estimate log–rank regressions of the form $\log(\mathrm{rank}) = a + b\cdot[-\log(\mathrm{stars})]$ on binned medians, where the slope $b$ maps directly to the implied Pareto exponent under the model’s canonical orientation. Columns progressively restrict the sample to increasingly elite repositories (top 75\%, 50\%, 10\%, and 5\% by rank).

\begin{table}[t]
\centering
\caption{GitHub stars: tail regressions of log rank on log stars}
\label{tab:stars-tail-reg}
\small

\begingroup
\centering
\begin{tabular}{lccccc}
   \tabularnewline \midrule \midrule
   Dependent Variable: & \multicolumn{5}{c}{log(rank\_mid)}\\
   Model:                 & (1)          & (2)          & (3)          & (4)          & (5)\\  
   \midrule
   \emph{Variables}\\
   Constant               & 16.5$^{***}$ & 18.7$^{***}$ & 23.3$^{***}$ & 27.8$^{***}$ & 25.9$^{***}$\\   
                          & (0.067)      & (0.108)      & (0.139)      & (0.418)      & (0.606)\\   
   Slope: $-\log(stars)$  & 1.08$^{***}$ & 1.33$^{***}$ & 1.80$^{***}$ & 2.24$^{***}$ & 2.07$^{***}$\\   
                          & (0.009)      & (0.013)      & (0.014)      & (0.038)      & (0.054)\\   
   \midrule
   \emph{Fit statistics}\\
   Observations           & 581          & 436          & 291          & 59           & 30\\  
   R$^2$                  & 0.96160      & 0.96240      & 0.98178      & 0.98392      & 0.98161\\  
   \midrule \midrule
   \multicolumn{6}{l}{\emph{IID standard-errors in parentheses}}\\
   \multicolumn{6}{l}{\emph{Signif. Codes: ***: 0.01, **: 0.05, *: 0.1}}\\
\end{tabular}
\par\endgroup

\begin{flushleft}
\footnotesize
\emph{Notes:} Each column reports a binned log--rank regression for GitHub repositories, estimated on median stars within logarithmic rank bins. The dependent variable is $\log(\mathrm{rank}_{\text{mid}})$ and the regressor is $-\log(\mathrm{stars}_{\text{med}})$, so the reported coefficient is the (positive) slope in the canonical Zipf orientation. Column~1 uses the full sample; Columns~2--5 restrict to the top 75\%, 50\%, 10\%, and 5\% of repositories by rank, respectively. Standard errors are iid. The increasing slope in the extreme tail indicates stronger concentration among the most-starred repositories.
\end{flushleft}
\end{table}
\section{A Model of the Open Source Software
Ecosystem}\label{a-model-of-the-open-source-software-ecosystem}

This section develops a simple model of the open source software (OSS)
ecosystem. We begin with a baseline economy without vibe coding: users
choose among shared packages that differ in quality, while developers
decide whether to release their projects as OSS after learning realized
quality. The equilibrium jointly determines software variety and
quality---captured by the mass of shared packages and the endogenous
distribution of quality---together with market-clearing relationships
that link user adoption to developer payoffs. We then introduce vibe
coding as a technology shock that changes how users access OSS and how
developers appropriate value, and we trace how this new usage technology
feeds back into sharing and entry decisions. 

\subsection{Decision Makers and Environment}
\subsubsection{Users} There is a unit mass of final users with
heterogeneous software needs. Each user can choose from a mass \(m_s\)
of shared software packages. Because software is nonrival, the same
package can be used by many users. Package \(k\) has quality \(q_k\) and
delivers flow utility \(u q_k\) to the representative user. The
equilibrium objects \(m_s\), the distribution of qualities, and the
utility shifter \(u\) are taken as given by individual users. Quality is
heterogeneous with distribution \(\mu(q)\).

The utility of user \(i\in[0,1]\) from using package \(k\in[0,m_s]\) is
\begin{equation}
U_{ik} = \nu_{ik} q_k u, \qquad \nu_{ik} \stackrel{\text{i.i.d.}}{\sim} \text{Fr\'{e}chet}\left(\sigma, 1/\Gamma(1-1/\sigma)\right), \label{eq:outer-frechet}
\end{equation} The idiosyncratic preference shock \(\nu_{ik}\) captures
heterogeneity in user needs: a given package may be a good match for some tasks but not others. The shape parameter \(\sigma > 1\)
governs substitutability across packages; the scale
parameter normalizes \(\mathbb{E}[\nu_{ik}] = 1\).

Each user chooses the package \(k\) that maximizes \(U_{ik}\).
Integrating over the preference shocks implies that a given package is
chosen with probability \[
\Pr(k) = \frac{q_k^\sigma}{m_s\int q_j^\sigma \, d\mu(j)}.
\] The adoption probability is increasing in quality. If all packages
have the same quality, users randomize uniformly and \(\Pr(k) = 1/m_s\).
We denote \[
\bar q = \left[\int q_j^\sigma \, d\mu(j)\right]^{1/\sigma}
\] the average quality of distributed software.

Using standard discrete-choice results
\citep{anderson-depalma-thisse-1992}, the expected utility of a user can
be written as \[
\mathcal U = \bar q u m_s^{1/\sigma}.
\] Expected utility is increasing in average software quality \(\bar q\)
and in the utility shifter \(u\), which we later allow to change with
vibe coding. It is also increasing in the mass of available packages
\(m_s\): holding the quality distribution fixed, a larger choice set
raises welfare through a love-of-variety (or better-match) effect, with
curvature governed by \(\sigma\).

A package of quality \(q_k\) is chosen by
\(n(q_k) = \Pr(k) = q_k^\sigma/(m_s \bar q^\sigma)\) users. The user
base is increasing in own quality, but decreasing in the number of
competing packages and in their average quality.

\subsubsection{Developers}\label{developers}

Each package is developed by a single developer (we abstract from
collaboration; for large projects, interpret the team as a single
decision maker). Development proceeds in two stages. First, the
developer incurs an up-front cost to design, build, and test the
developer incurs an up-front cost \(\Phi\) to design, build, and test the
software before its quality is known. This assumption, standard in
heterogeneous-firm models \citep{melitz-2003}, captures that
software quality is hard to predict ex ante
\citep[page 136]{jones-2020-eseur}. Second, after observing
realized quality, the developer decides whether to release the package
as open source by paying a fixed sharing cost \(\tau\), which captures
packaging, documentation, and ongoing maintenance.

Developers share because a larger user base generates private returns---reputation, career benefits, and monetization opportunities. We capture these
channels by allowing developer payoff to depend on the
number of users: \[
\Pi(q_k) = \pi n(q_k) = \pi q_k^\sigma/(m_s \bar q^\sigma),
\] where \(\pi\) is the reward per user. We allow \(\pi\) to depend on
the usage technology and, in particular, on whether users interact
directly with the project or through vibe coding; we specify this
dependence below.

Project quality is heterogeneous. For tractability, assume each
developer draws realized project quality \(q \ge 1\) i.i.d. from a
Pareto distribution with shape parameter \(\gamma>0\) and normalized
scale, \[
\Pr(q>x) = x^{-\gamma}, \qquad x\ge 1,
\] so that the density is \(\gamma q^{-\gamma-1}\). The Pareto
assumption captures the empirical right tail of software project
outcomes and delivers closed-form expressions for equilibrium objects.

After observing \(q\), the developer decides whether to release the
project as OSS. Sharing requires paying a fixed cost \(\tau\) and yields
payoff \(\Pi(q)\). Because user adoption is increasing in project
quality, developer payoff is increasing in \(q\). In particular, \[
\Pi(q) = \pi n(q) = \frac{\pi q^\sigma}{m_s \bar q^\sigma}.
\] A developer shares if and only if the payoff from doing so covers the
sharing cost, \[
\Pi(q) \ge \tau.
\] Since \(\Pi(q)\) is increasing in \(q\), the sharing decision is
characterized by a cutoff \(q_0\) defined by the indifference condition
\(\Pi(q_0)=\tau\): projects with \(q\ge q_0\) are shared, while projects
with \(q<q_0\) are not.

The cutoff \(q_0\) is endogenous. It depends on the equilibrium objects
\(m_s\) and \(\bar q\) that determine user adoption and, hence,
developer payoffs; at the same time, both \(m_s\) and \(\bar q\) depend
on \(q_0\) through the selection of which quality draws are shared.

Given a cutoff \(q_0\), the distribution of quality among shared
projects is a truncated Pareto. Under the quality aggregator defined
above, average quality satisfies \[
\bar q = \Lambda q_0,
\] where
\(\Lambda = \left(\frac{\gamma}{\gamma-\sigma}\right)^{1/\sigma} > 1\).
The condition \(\gamma>\sigma\) ensures that the moment determining
\(\bar q\) is finite. If there is a mass \(m\) of developers who have
built a software package, the mass of shared software is \[
m_s = m\Pr(q\ge q_0) = m q_0^{-\gamma}.
\] A higher cutoff reduces the fraction of projects that are shared
(lower \(m_s\)), but raises the average quality of shared OSS (higher
\(\bar q\)).

Imposing the sharing indifference condition \(\Pi(q_0)=\tau\) and
substituting \(\bar q=\Lambda q_0\) and \(m_s=m q_0^{-\gamma}\) yields
the equilibrium cutoff \[
q_0 = (\tau m \Lambda^\sigma /\pi)^{1/\gamma}.
\] This implies equilibrium average quality \[
\bar q = \Lambda^{1+\sigma/\gamma}(\tau m/\pi)^{1/\gamma}
\] and equilibrium variety \[
m_s = \frac{\pi}{\tau \Lambda^{\sigma}},
\] so that the mass of shared software increases in the reward per user
\(\pi\) and decreases in the sharing cost \(\tau\).

Notably, \(m_s\) is independent of the mass of produced projects \(m\).
Entry increases the pool of potential projects, but it also intensifies
competition for users, reducing the payoff from sharing and raising the
cutoff \(q_0\); with Pareto tails, these two forces exactly offset in
equilibrium.

Substituting equilibrium objects back into developer payoff gives \[
\Pi(q) = \frac{\pi q^\sigma}{m_s \bar q^\sigma} = q^\sigma \Lambda^{-\sigma^2/\gamma} \tau \left(\frac{\pi}{\tau m}\right)^{\sigma/\gamma} = \tau\left(\frac{q}{q_0}\right)^{\sigma}.
\] Developer payoff is increasing in quality \(q\) and in the reward per
user \(\pi\), and decreasing in the mass of competing projects \(m\). It
is also increasing in the sharing cost \(\tau\) because \(\gamma>\sigma\):
higher sharing costs discourage marginal developers from sharing,
reducing competition for those who do share (formally,
\(\Pi(q)\propto \tau^{1-\sigma/\gamma}\)).

Finally, before drawing quality, a developer's expected net payoff from
entering the development stage (gross of the common up-front cost
\(\Phi\)) is \[
\mathbb E[\max\{\Pi(q)-\tau,0\}] = \int_{q=q_0}^\infty [\Pi(q)-\tau]\,\gamma q^{-\gamma-1}\, d q = \frac{\sigma}{\gamma - \sigma} \cdot \frac{\pi}{m \Lambda^\sigma}.
\] Low-quality draws \(q<q_0\) are not shared and therefore earn zero.

Free entry requires that the expected payoff from developing software
equals the up-front cost: \[
\Phi = \frac{\sigma}{\gamma-\sigma}\frac{\pi}{m\Lambda^\sigma}.
\] Rearranging pins down the equilibrium mass of developers: \[
m =\frac{\sigma}{\gamma-\sigma}\frac{\pi}{\Phi\Lambda^\sigma}. 
\] Entry is increasing in the reward per user \(\pi\) (larger market)
and decreasing in the development cost \(\Phi\) (higher barrier). Entry
is also decreasing in the quality aggregator \(\Lambda\): fatter-tailed
quality distributions (lower \(\gamma\), higher \(\Lambda\)) imply that
high-quality projects capture more of the market, making entry less
attractive for marginal developers.

\subsection{Software Production
Function}\label{software-production-function}

The cost of developing software is endogenous. Developers rely on their
own time and effort, but also use existing software tools to build new
ones. When software quality and variety increase, development becomes
cheaper---a feedback we call the \emph{software-begets-software} effect.
Importantly, humans remain essential: they design software architecture,
make product decisions, and oversee implementation, even when AI agents
write substantial portions of the code. Survey evidence confirms that
experienced developers use AI agents as productivity tools while
retaining control over design and implementation
\citep{huang-2025-control}.

We model software production with a Cobb-Douglas technology: elasticity
\(\beta \in (0,1)\) to existing software and \(1-\beta\) to labor.
Software can substitute for human effort, but not completely. The parameter \(\beta\)
governs the strength of the software-begets-software feedback.

Developers value
software quality and variety in the same way as final users---their
productive value from the software ecosystem is \[
\mathcal U = \bar q u m_s^{1/\sigma} = u\Lambda^{\sigma/\gamma}m^{1/\gamma}{\pi^{1/\sigma-1/\gamma}}{\tau^{1/\gamma-1/\sigma}}.
\] This symmetry is natural: developers are themselves users of existing
packages when building new software. The same functional form also fits the final-user side of the market. In practice, the ``consumer'' in the model is often a programmer assembling a custom, one-off solution from existing OSS packages (for example, a script, a data pipeline, or an internal tool). We model this assembly problem in reduced form with the same aggregator $\mathcal U = \bar q\,u\,m_s^{1/\sigma}$ because what matters for both usage and production is access to high-quality packages, a productivity shifter $u$, and variety $m_s$. Final users, however, do not contribute new OSS varieties: they take $(m_s,\bar q)$ as given and choose which package to use and whether to use it directly or via vibe coding. We also normalize away an explicit labor choice for final users by treating their time/effort constraint as a scale parameter; this affects the level of utility but not the adoption decision that pins down the vibe-coding share $v$.

The software development cost is
therefore \[
\Phi = \kappa^{1-\beta}\mathcal U^{-\beta} = 
\kappa^{1-\beta}
u^{-\beta}\Lambda^{-\beta\sigma/\gamma}m^{-\beta/\gamma}
\left(\frac\tau{\pi}\right)^{\beta/\sigma-\beta/\gamma}.
\] Development cost is increasing in labor cost \(\kappa\) and
decreasing in software quality \(\bar q\), productivity \(u\), and
variety \(m_s\). This expression follows from cost minimization under
Cobb-Douglas: minimizing \(\kappa L + p_S S\) subject to
\(L^{1-\beta}S^\beta = 1\) yields unit cost proportional to
\(\kappa^{1-\beta}p_S^\beta\); interpreting the ``price'\,' of software
as the inverse of its productivity \(\mathcal{U}^{-1}\) gives the stated
form.

The productivity shifter \(u\) will capture the effect of vibe coding on
development: AI assistance makes any given package easier to use and
faster to integrate. Because \(u\) enters with exponent \(-\beta\), the
cost reduction from higher \(u\) is proportional to the software share
\(\beta\). When software accounts for a larger share of production
(\(\beta\) higher), improvements in software productivity translate into
larger cost reductions.

\begin{assumption}[Parameter Restrictions]\label{ass:parameters}
Quality dispersion exceeds variety substitutability:
$$
\gamma > \sigma.
$$
\end{assumption}

\subsection{Equilibrium}\label{equilibrium}

We now define equilibrium in the baseline economy (without vibe coding,
so \(u\) is a fixed productivity parameter and all users interact
directly with OSS).

\begin{definition}[Baseline Equilibrium]\label{def:baseline-eq}
A \emph{baseline equilibrium} is a tuple $(m, q_0, m_s, \bar q)$ with $m > 0$, $q_0 \ge 1$, $m_s > 0$, and $\bar q > 0$ such that:
\begin{enumerate}
\item \textbf{(User optimality)} Each user chooses package $k$ to maximize $U_{ik} = \nu_{ik} q_k u$, implying adoption probability $\Pr(k) = q_k^\sigma / (m_s \bar q^\sigma)$.
\item \textbf{(Sharing cutoff)} A developer shares if and only if $\Pi(q) \ge \tau$, where $\Pi(q) = \pi q^\sigma / (m_s \bar q^\sigma)$. The marginal project satisfies $\Pi(q_0) = \tau$.
\item \textbf{(Quality aggregation)} Average quality among shared projects is $\bar q = \Lambda q_0$ with $\Lambda = \left(\frac{\gamma}{\gamma - \sigma}\right)^{1/\sigma}$.
\item \textbf{(Mass of shared software)} Given mass $m$ of developed projects, $m_s = m q_0^{-\gamma}$.
\item \textbf{(Free entry)} Developers enter until expected payoff equals development cost:
$$
\frac{\sigma}{\gamma - \sigma} \cdot \frac{\pi}{m \Lambda^\sigma} = \Phi,
$$
where $\Phi = \kappa^{1-\beta}u^{-\beta}\Lambda^{-\beta\sigma/\gamma}m^{-\beta/\gamma}\left(\frac\tau{\pi}\right)^{\beta/\sigma-\beta/\gamma}$.
\end{enumerate}
\end{definition}

Conditions (2)--(4) jointly determine \(q_0\), \(\bar q\), and \(m_s\)
as functions of \(m\). Condition (5) then pins down the equilibrium mass
of developers \(m\).

The software-begets-software feedback creates a magnification. When
entry increases, quality rises, which lowers development costs, which
encourages more entry.

\begin{proposition}[Equilibrium Existence and Uniqueness]\label{prop:existence}
Under Assumption~\ref{ass:parameters}, a unique baseline equilibrium exists. The equilibrium is characterized by:
\begin{align}
q_0 &= \left(\frac{\tau m \Lambda^\sigma}{\pi}\right)^{1/\gamma}, \label{eq:q0-eq}\\
\bar{q} &= \Lambda q_0 = \Lambda^{1+\sigma/\gamma}\left(\frac{\tau m}{\pi}\right)^{1/\gamma}, \label{eq:qbar-eq}\\
m_s &= \frac{\pi}{\tau \Lambda^\sigma}, \label{eq:ms-eq}
\end{align}
where $m$ solves the reduced-form free entry condition
\begin{equation}
m^{1 - \beta/\gamma} = \frac{\sigma}{\gamma - \sigma} \cdot \Lambda^{\beta\sigma/\gamma - \sigma} \cdot \kappa^{\beta-1} \cdot u^{\beta} \cdot \pi^{1 + \beta/\sigma - \beta/\gamma} \cdot \tau^{\beta/\gamma - \beta/\sigma}. \label{eq:m-reduced}
\end{equation}
\end{proposition}

\begin{proof}
See Appendix~\ref{app:proofs}.
\end{proof}

Define \(\eta \equiv 1 - \beta/\gamma > 0\). Taking logs of
\eqref{eq:m-reduced} and differentiating:

\begin{corollary}[Comparative Statics]\label{cor:comp-statics}
Equilibrium entry $m$ satisfies:
\begin{align}
\frac{\partial \log m}{\partial \log \pi} &= \frac{1 + \beta/\sigma - \beta/\gamma}{\eta} > 1, \label{eq:dm-dpi}\\
\frac{\partial \log m}{\partial \log u} &= \frac{\beta}{\eta} > 0, \label{eq:dm-du}\\
\frac{\partial \log m}{\partial \log \kappa} &= -\frac{1-\beta}{\eta} < 0, \label{eq:dm-dkappa}\\
\frac{\partial \log m}{\partial \log \tau} &= \frac{\beta/\gamma - \beta/\sigma}{\eta}. \label{eq:dm-dtau}
\end{align}
\end{corollary}

The first elasticity exceeds one:
\(1 + \beta/\sigma - \beta/\gamma > 1 - \beta/\gamma\) because
\(\beta/\sigma > 0\). This reflects the software-begets-software
feedback: higher \(\pi\) raises entry, which improves quality and
variety, which lowers costs and further encourages entry. The second
elasticity \(\beta/\eta\) is positive; since \(\eta < 1\), we have
\(\beta/\eta > \beta\), amplifying the direct productivity effect.

Entry increases in developer reward \(\pi\) and software productivity
\(u\), and decreases in labor cost \(\kappa\). The effect of sharing
cost \(\tau\) depends on parameters. The intuition for each effect:

\begin{itemize}
\item \emph{Reward ($\pi$)}: Higher reward per user makes sharing more attractive, drawing in marginal developers. The elasticity exceeds one because of the software-begets-software feedback: more entry improves the ecosystem, lowering costs and inducing further entry.
\item \emph{Productivity ($u$)}: Higher software productivity lowers development costs, encouraging entry. The elasticity $\beta/\eta > \beta$ reflects amplification through the feedback loop.
\item \emph{Labor cost ($\kappa$)}: Higher wages raise development costs, discouraging entry. The magnitude is smaller when software accounts for a larger share of production.
\item \emph{Sharing cost ($\tau$)}: Higher $\tau$ discourages sharing (reducing $m_s$), but it also raises the quality cutoff (improving $\bar{q}$). Because $\gamma > \sigma$, the quality effect dominates and entry rises with $\tau$.
\end{itemize}

\begin{corollary}[Equilibrium Quality and Welfare]\label{cor:quality-welfare}
Average quality and user welfare in equilibrium are:
\begin{align}
\bar{q} &= \Lambda^{1+\sigma/\gamma} \left(\frac{\tau}{\pi}\right)^{1/\gamma} m^{1/\gamma}, \label{eq:qbar-m}\\
\mathcal{U} &= \bar{q} \, u \, m_s^{1/\sigma} = \Lambda^{\sigma/\gamma} \left(\frac{\pi}{\tau}\right)^{1/\sigma - 1/\gamma} u \, m^{1/\gamma}. \label{eq:U-m}
\end{align}
Both are increasing in $m$. Hence any parameter change that raises entry also raises average quality and user welfare.
\end{corollary}

Define the \emph{welfare exponent} $\omega \equiv 1/(\gamma\eta) = 1/(\gamma - \beta)$. Substituting the equilibrium value of $m$ from \eqref{eq:m-reduced}, both $\bar{q}$ and $\mathcal{U}$ can be expressed purely in terms of exogenous parameters. In particular, both are increasing in software productivity $u$ (with exponent $\beta\omega$ for $\bar{q}$ and $1 + \beta\omega$ for $\mathcal{U}$) and decreasing in labor cost $\kappa$ (with exponent $-(1-\beta)\omega$).

The share of user utility captured by developers is the ratio of development cost to welfare. Empirically, this ratio is approximately 0.001---development costs are about 0.1\% of user utility. The small ratio reflects two forces: high productivity raises user welfare faster than development costs, and competition among developers dissipates rents. Users receive roughly 1000 times more value than developers spend creating that value---a consequence of software's nonrivalry and the scale of adoption.

\subsubsection{First-Best Allocation}\label{first-best-allocation}

The comparative statics show that policy can affect entry and welfare.
But is the market providing the right amount of open source software? To
answer this, we characterize the first-best allocation and compare it to
equilibrium.

A social planner maximizes aggregate welfare, defined as user utility
minus total development costs. Developer surplus is zero under free entry, so aggregate welfare reduces to user utility minus development costs.

\begin{definition}[First-Best]\label{def:first-best}
The \emph{first-best allocation} solves
$$
\max_{m \ge 0} \; \mathcal{U}(m) - m \cdot \Phi(m),
$$
where $\mathcal{U}(m)$ and $\Phi(m)$ are given by the equilibrium expressions \eqref{eq:U-m} and the development cost evaluated at equilibrium objects. 
\end{definition}

\begin{theorem}[Underprovision of Entry]\label{thm:underprovision-baseline}
Suppose developers capture less than the full social value of a user, i.e., $\pi < \bar{q} u$. Then equilibrium entry is below the first-best level: $m^{eq} < m^{fb}$. Equivalently, there are too few shared software packages in equilibrium.
\end{theorem}

\begin{proof}
See Appendix~\ref{app:proofs}.
\end{proof}

The intuition is straightforward: the social marginal benefit of entry
includes the full contribution to \(\mathcal{U}\), while the private
marginal benefit is only \(\pi\) per user. When \(\pi < \bar{q}u\), the
private return understates the social return, so developers enter less
than is socially optimal.

The condition \(\pi < \bar{q}u\) is the natural case: developers capture
a share of user surplus through reputation, donations, or monetization,
but not the full amount. The theorem implies that policies raising
developer appropriability (higher \(\pi\)) or subsidizing entry (lower
\(\Phi\)) move the economy toward the first-best. Real-world examples include GitHub Sponsors (raising \(\pi\)) and foundation grants for infrastructure projects (lowering \(\Phi\)) \citep{boysel-nagle-carter-hermansen-crosby-luszcz-lincoln-yue-hoffmann-staub-2024-oss-funding}.

These baseline results set the stage for analyzing vibe coding. The key
question is how AI-assisted development---which changes both \(u\) and
\(\pi\)---affects the equilibrium provision of open source software.

\subsection{Vibe Coding}\label{vibe-coding}

We now extend the baseline model to incorporate vibe coding. The
extension adds an inner nest to the user's problem: after choosing a
package, the user decides how to use it. This choice affects both user
utility and developer revenue, because vibe-coded usage does not
generate the direct engagement that developers monetize.

Once a user has chosen a particular package \(k\), she can choose
between two usage modes:

\begin{enumerate}
\item \emph{Direct usage}: the user engages directly with the open source project---reading documentation, visiting the project website, potentially discovering the developer's commercial offerings
\item \emph{Vibe-coded usage}: the user interacts with the software through an AI assistant that has been trained on the software's documentation and codebase
\end{enumerate}

We normalize the productivity of direct usage of a software package of quality $q$ to $q$. Vibe-coded usage has
productivity \(q \zeta \geq 0\), which captures how well AI assistants can
help users accomplish their goals. When \(\zeta = 0\), vibe coding is
not available; when \(\zeta = 1\), vibe coding is equally productive as
direct usage; when \(\zeta > 1\), vibe coding is more productive. Both
options are subject to idiosyncratic taste shocks that capture
heterogeneity in user preferences over interaction modes:
\begin{equation}
u_k = \max\{q \varepsilon_{k1}, \zeta q \varepsilon_{k2}\}, \qquad \varepsilon_{ki} \stackrel{\text{i.i.d.}}{\sim} \text{Fr\'{e}chet}\left(\theta, 1/\Gamma(1-1/\theta)\right), \label{eq:inner-frechet}
\end{equation} where \(\theta > 1\) is the elasticity of substitution
between direct and vibe-coded usage. A higher \(\theta\) means the two
modes are closer substitutes, so users are more responsive to
differences in productivity.

The choice probability follows from the same discrete-choice logic as
the selection of a software package. Since both utilities scale with package quality \(q\),
quality cancels in the choice probability. The probability that a user
chooses vibe coding is: \[
v = \frac{\zeta^\theta}{1+\zeta^\theta}.
\] This logistic function has an intuitive interpretation. When
\(\zeta < 1\), vibe coding is less productive than direct usage, and
fewer than half of users adopt it. When \(\zeta = 1\), users are
indifferent on average, and exactly half choose vibe coding. As
\(\zeta\) rises above 1, vibe coding becomes more attractive, and
adoption rises rapidly. Figure \ref{fig:vibe-adoption} illustrates this
S-curve for \(\theta = 3.5\).


\begin{figure}[htbp]
\centering
\begin{tikzpicture}[scale=1.2]
    \draw[->] (0,0) -- (4.5,0) node[right] {$\zeta$};
    \draw[->] (0,0) -- (0,3.5) node[above] {$v$};
    
    \node[below] at (2,0) {$1$};
    \node[left] at (0,2.5) {$1$};
    \node[left] at (0,1.25) {$\frac{1}{2}$};
    
    \draw (2,-0.05) -- (2,0.05);
    \draw (-0.05,2.5) -- (0.05,2.5);
    \draw (-0.05,1.25) -- (0.05,1.25);
    
    \draw[dashed, gray] (2,0) -- (2,2.5);  
    \draw[dashed, gray] (0,2.5) -- (4.2,2.5);  
    \draw[dashed, gray] (0,1.25) -- (4.2,1.25);  
    
    \draw[thick, blue, domain=0:4.2, samples=200] 
        plot (\x, {2.5 * ((\x/2)^3.5) / (1 + (\x/2)^3.5)});
    
    \fill[blue] (2, 1.25) circle (2pt);
    
    \node[right, font=\small] at (3.2, 1.8) {$v = \frac{\zeta^\theta}{1+\zeta^\theta}$};
    \node[right, font=\small] at (3.2, 1.4) {$(\theta = 3.5)$};
    
\end{tikzpicture}
\caption{Vibe coding adoption as a function of AI capability $\zeta$. When $\zeta < 1$, vibe coding is less productive than direct usage and adoption is below 50\%. As $\zeta$ crosses 1, adoption rises rapidly due to the large elasticity parameter $\theta = 3.5$. The steep S-curve illustrates how small improvements in AI capability can trigger rapid shifts in usage patterns.}
\label{fig:vibe-adoption}
\end{figure}
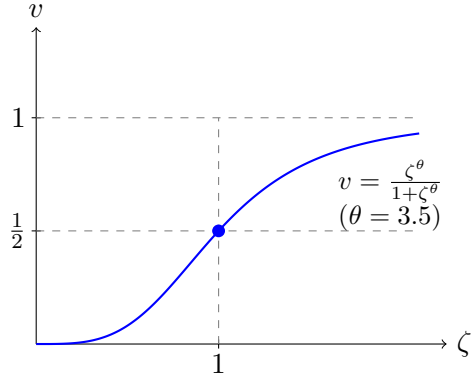

The parameter \(\theta\) governs how sensitive adoption is to changes in
AI capability. When \(\theta\) is large, small improvements in \(\zeta\)
around the threshold \(\zeta = 1\) trigger rapid swings in adoption. The
empirical evidence is consistent with large \(\theta\): AI-assisted
coding went from near-zero to 29 percent of Python functions within two
years \citep{daniotti-2025-ai-code}, and at Anthropic ``70, 80, 90
percent of the code is written by Claude'' \citep{amodei-2025}.

\begin{assumption}[Substitutability Ordering]\label{ass:substitutability}
Usage modes are closer substitutes than software packages:
$$
\theta > \sigma.
$$
\end{assumption}

This assumption is empirically plausible. The parameter \(\sigma\)
governs substitution across software packages---how easily users switch
from one library to another. The parameter
\(\theta\) governs substitution between direct and vibe-coded usage of
the \emph{same} package. Since both modes deliver the same underlying
functionality, they are closer substitutes than entirely different
software packages. Empirically, the rapid adoption of vibe coding---from
near-zero to majority usage within two years---is consistent with high
\(\theta\): users switch easily once AI capability crosses the
threshold. 

The expected utility from a package of quality \(q\), taking the max
over usage modes, is: \begin{equation}
\mathbb{E}[u_k] = q(1 + \zeta^{\theta})^{1/\theta} = q(1-v)^{-1/\theta}. \label{eq:expected-utility-package}
\end{equation} Define \(u \equiv (1-v)^{-1/\theta}\) as the
per-unit-quality utility multiplier from having the vibe coding option.
This multiplier exceeds 1 whenever \(v > 0\): the option to choose
between direct and vibe-coded usage always raises expected utility. The
gain is larger when \(\zeta\) is higher (better AI) or when \(\theta\)
is smaller (modes are more differentiated, so the option value is
larger).

The formula \((1-v)^{-1/\theta}\) parallels expressions in the trade
literature for welfare gains from imported varieties \citep{arkolakis-costinot-rodriguez-clare-2012,halpern-koren-szeidl-2015}: access to an additional option raises welfare by an amount that depends on the usage share and
the elasticity of substitution.

The utility multiplier \(u\) affects both final users and developers.
For users, higher \(u\) raises expected utility directly through
\eqref{eq:expected-utility-package}. For developers, higher \(u\) lowers
development cost \(\Phi\) because software is an input to its own
production. When AI capability \(\zeta\) is high, the vibe coding share
\(v\) is high, \(u = (1-v)^{-1/\theta}\) is large, and developers can
build new software at lower cost. This encourages entry.

Two parameters govern the strength of this feedback. When
\(\theta\) is large, direct and vibe-coded usage are close substitutes, so
the utility gain \((1-v)^{-1/\theta}\) is modest even for high \(v\).
When \(\beta\) is small, software contributes little to its own production, and the cost reduction \(\Phi \propto u^{-\beta}\) is correspondingly weak.

\subsubsection{Short-Run Effect: Developer Adoption
Only}\label{short-run-effect-developer-adoption-only}

In the short run, developers adopt AI assistance before final users
shift to vibe-coded consumption. Users still interact directly with open
source projects---visiting documentation, discovering developers'
commercial offerings, and generating the engagement that underlies
\(\pi\).

\begin{definition}[Short-Run Equilibrium]\label{def:short-run-eq}
A \emph{short-run equilibrium with vibe coding} is a baseline equilibrium (Definition~\ref{def:baseline-eq}) with the following modifications for some $v \in (0,1)$: developers benefit from AI-assisted development, so the productivity parameter in the development cost function is $u = (1-v)^{-1/\theta}$; users still interact directly with OSS, so the utility multiplier in their welfare remains $1$; and the reward per user $\pi$ remains at its baseline value $\bar{\pi}$.
\end{definition}

\begin{theorem}[Short-Run Effect of Vibe Coding]\label{thm:short-run}
The short-run equilibrium (Definition~\ref{def:short-run-eq}) satisfies:
\begin{align}
\frac{m}{m_0} &= (1-v)^{-\beta/(\theta\eta)}>1, \label{eq:m-short}\\
\frac{m_s}{m_{s,0}} &= 1, \label{eq:ms-short}\\
\frac{\bar{q}}{\bar{q}_0} &= (1-v)^{-\beta/(\theta\gamma\eta)}>1, \label{eq:qbar-short}\\
\frac{\mathcal{U}}{\mathcal{U}_0} &= (1-v)^{-\beta/(\theta\gamma\eta)}>1, \label{eq:U-short}
\end{align}
where $\eta = 1 - \beta/\gamma$ and subscript $0$ denotes baseline equilibrium values (computed at $v=0$). This notation should not be confused with the quality cutoff $q_0$, which is an endogenous variable defined by the sharing indifference condition.

All ratios except $m_s/m_{s,0}$ exceed 1 for $v \in (0,1)$. Vibe coding raises entry, quality, and welfare in the short run, though the welfare gain is modest: users benefit only from quality selection, not from AI-assisted consumption.
\end{theorem}

\begin{proof}
From \eqref{eq:m-reduced}, $m^\eta \propto u^{\beta}$ (holding $\pi$ fixed). Taking logs and differentiating: $\partial \log m / \partial \log u = \beta/\eta$. Substituting $u = (1-v)^{-1/\theta}$ for developer productivity gives \eqref{eq:m-short}. Equation \eqref{eq:ms-eq} shows $m_s = \pi/(\tau\Lambda^\sigma)$, independent of $u$, yielding \eqref{eq:ms-short}. From \eqref{eq:qbar-eq}, $\bar{q} \propto m^{1/\gamma}$, giving \eqref{eq:qbar-short}. Finally, $\mathcal{U} = \bar{q} \cdot m_s^{1/\sigma}$ (with user utility multiplier equal to 1 since users interact directly) and $m_s$ unchanged implies $\mathcal{U}/\mathcal{U}_0 = \bar{q}/\bar{q}_0$, yielding \eqref{eq:U-short}.
\end{proof}

A corollary of the short-run results is that the share of projects that
get shared, \(m_s/m\), decreases with vibe coding. Entry rises
(\(m \uparrow\)) while variety is unchanged (\(m_s\) constant), so a
smaller fraction of projects clear the quality threshold for sharing.
The model thus predicts an increase in low-quality, unshared
projects---software that is created but never distributed. This
prediction resonates with practitioner accounts of ``AI slop'': the
proliferation of AI-generated code pushed to public repositories that
attracts no users and serves no purpose beyond its creator's immediate
needs. The curl maintainer reports that 20 percent of security reports
in 2025 are AI-generated slop, consuming hours of volunteer time to
review fabricated vulnerabilities \citep{stenberg-2025}. Meanwhile,
aggregate git pushes to public repositories continue rising
\citep{github-innovation-graph-2025}, consistent with more projects
being created even as quality thresholds tighten.

The short-run analysis holds \(\pi\) fixed, but this assumption cannot
persist. Equation \eqref{eq:inner-frechet} and
Figure \ref{fig:vibe-adoption} show that the vibe coding
share \(v\) rises steeply with AI capability \(\zeta\). When OSS
business models monetize visibility from direct users---through
documentation traffic, consulting leads, or reputation---developer
reward falls in proportion to \((1-v)\). As vibe coding spreads, the
same library may see rising downloads but falling engagement. This
erosion of the reward channel is the key mechanism behind the long-run
results.

\subsubsection{Long-Run Effect: Traditional Business
Models}\label{long-run-effect-traditional-business-models}

\begin{assumption}[Traditional Business Model]\label{ass:traditional}
Developer reward is proportional to direct user engagement:
$$
\pi = \bar{\pi}(1-v),
$$
where $\bar{\pi}$ is the baseline reward per user and $(1-v)$ is the share of users who interact directly with the project.
\end{assumption}

\begin{theorem}[Long-Run Effect of Vibe Coding]\label{thm:long-run}
Under Assumptions~\ref{ass:traditional} and \ref{ass:substitutability}, the long-run equilibrium satisfies:
\begin{align}
\frac{m}{m_0} &= (1-v)^{1 + \beta(1/\sigma - 1/\theta)/\eta}<1, \label{eq:m-long}\\
\frac{m_s}{m_{s,0}} &= (1-v)<1, \label{eq:ms-long}\\
\frac{\bar{q}}{\bar{q}_0} &= (1-v)^{\beta(1/\sigma - 1/\theta)/(\gamma\eta)}<1, \label{eq:qbar-long}\\
\frac{\mathcal{U}}{\mathcal{U}_0} &= (1-v)^{1/\sigma - 1/\theta + \beta(1/\sigma - 1/\theta)/(\gamma\eta)}. \label{eq:U-long}
\end{align}
where $\eta = 1 - \beta/\gamma$. Entry, variety, and quality all fall.
\end{theorem}

\begin{proof}
From \eqref{eq:m-reduced}, $m^{\eta} \propto \pi^{a_\pi} u^{\beta}$ where $a_\pi = 1 + \beta/\sigma - \beta/\gamma$. Substituting $\pi = \bar{\pi}(1-v)$ and $u = (1-v)^{-1/\theta}$:
$$
m^{\eta} \propto (1-v)^{a_\pi} (1-v)^{-\beta/\theta} = (1-v)^{a_\pi - \beta/\theta}.
$$
The exponent $a_\pi - \beta/\theta = 1 + \beta/\sigma - \beta/\gamma - \beta/\theta = \eta + \beta(1/\sigma - 1/\theta)$. Taking $\eta$-th roots gives \eqref{eq:m-long}. Equation \eqref{eq:ms-eq} with $\pi = \bar{\pi}(1-v)$ gives \eqref{eq:ms-long}. From \eqref{eq:qbar-eq}, $\bar{q} \propto (m/\pi)^{1/\gamma}$. The ratio $(m/\pi)/(m_0/\bar{\pi}) = (m/m_0)/(\pi/\bar{\pi}) = (1-v)^{\beta(1/\sigma - 1/\theta)/\eta}$, yielding \eqref{eq:qbar-long}. Finally, $\mathcal{U} = \bar{q} u m_s^{1/\sigma}$ gives
$$
\frac{\mathcal{U}}{\mathcal{U}_0} = \frac{\bar{q}}{\bar{q}_0} \cdot \frac{u}{u_0} \cdot \left(\frac{m_s}{m_{s,0}}\right)^{1/\sigma} = (1-v)^{\beta(1/\sigma - 1/\theta)/(\gamma\eta)} \cdot (1-v)^{-1/\theta} \cdot (1-v)^{1/\sigma},
$$
and collecting exponents yields \eqref{eq:U-long}.
\end{proof}

The long-run theorem reveals a horse race between two channels. The cost
channel (higher \(u\)) encourages entry with elasticity
\(\beta/(\theta\eta)\). The reward channel (lower \(\pi\)) discourages
entry with elasticity \((1 + \beta/\sigma - \beta/\gamma)/\eta > 1\).
The reward channel dominates: the net exponent on entry,
\(1 + \beta(1/\sigma - 1/\theta)/\eta\), is positive, so entry falls as vibe coding spreads.

The parameters \(\theta\), \(\sigma\), and \(\beta\) govern the
magnitude. When \(\theta\) is large, the productivity gain from vibe
coding is small (users are nearly indifferent between modes), so the
cost channel is weak. When \(\beta\) is small, software contributes
little to its own production, again weakening the cost channel. In both
cases, the reward channel dominates more strongly, and OSS provision
collapses faster.

The rapid empirical adoption of vibe coding---from near-zero to majority
usage within two years---tells us that either \(\theta\) or \(\zeta\)
(or both) must be large. From \(v = \zeta^\theta/(1+\zeta^\theta)\),
high \(v\) requires \(\zeta^\theta \gg 1\). But the welfare implications
differ sharply depending on which parameter drives adoption. If
\(\theta\) is large, direct and vibe-coded usage are close substitutes:
users switch easily, but the productivity gain \(u = (1-v)^{-1/\theta}\)
is modest. The cost channel is weak, and the reward channel
dominates---OSS collapses quickly. If instead \(\zeta\) is large (AI is
highly productive), the same adoption share \(v\) delivers a larger
productivity boost, strengthening the cost channel. Distinguishing these
cases empirically---through revealed preference in switching behavior or
direct productivity measurement---is crucial for predicting the long-run
trajectory of OSS provision.

\subsubsection{Alternative Business Models}\label{alternative-business-models}

The long-run analysis assumes developers monetize only through direct user engagement. This assumption abstracts from alternative business models that either monetize AI-mediated usage (for example through API fees charged to AI providers, telemetry-based attribution systems, or revenue-sharing agreements) or rely on revenue streams that are less tightly linked to how users access the software (enterprise licensing, foundation grants, or services sold to other developers).

The key observation is that per-user monetization can fall without destroying OSS: if vibe coding makes development cheaper, lower rewards can still support entry. However, monetization cannot fall too much. We now characterize the minimum monetization needed to sustain the pre-vibe-coding level of OSS entry.

Let $\pi$ denote per-user monetization when the vibe-coding adoption share is $v$, and let $\pi_0$ denote the corresponding per-user monetization in a baseline economy without vibe coding ($v=0$). We are agnostic about the specific revenue model that generates $\pi$; we only ask how large the ratio $\pi/\pi_0$ must be.

\begin{theorem}[Minimum Monetization Bound]\label{thm:min-monetization}
Under Assumptions~\ref{ass:parameters} and \ref{ass:substitutability}, define
$$
\omega \equiv \frac{1/\theta}{1/\beta + 1/\sigma - 1/\gamma} \in (0,1).
$$
Sustaining baseline entry ($m \ge m_0$) requires:
\begin{equation}\label{eq:min-monetization}
\frac{\pi}{\pi_0} \ge (1-v)^\omega.
\end{equation}
Equivalently, per-user monetization can decline by at most $1-(1-v)^\omega$.
\end{theorem}

The bound is tight because entry responds more elastically to monetization than to productivity. From Corollary~\ref{cor:comp-statics}, the elasticity of entry with respect to reward is $(1 + \beta/\sigma - \beta/\gamma)/\eta > 1$, while the elasticity with respect to the productivity shifter is $\beta/\eta < 1$. Their ratio is $1/\beta + 1/\sigma - 1/\gamma$, which is about $3.3$ under our calibration.

This asymmetry makes the quantitative implication stark. With $v=0.7$ and $\theta=3$, the productivity shifter is $u=(1-v)^{-1/\theta}=0.3^{-1/3}\approx 1.49$, which lowers development cost by a factor of $u^{-\beta}=1.49^{-1/3}\approx 0.88$, about a 12 percent reduction. Yet Theorem~\ref{thm:min-monetization} implies that per-user monetization can fall by at most $1-(1-v)^\omega \approx 11$ percent (since $(1-v)^\omega = 0.3^{0.1} \approx 0.89$). Under the traditional business model (Assumption~\ref{ass:traditional}), per-user monetization falls by 70 percent ($\pi/\pi_0 = 1-v = 0.3$), far beyond what this cost reduction can offset. The gap between a 70 percent revenue collapse and an 11 percent sustainable decline motivates alternative business models.

To connect the bound to concrete revenue models, consider a simple decomposition of per-user monetization into three components. A share $\alpha \in [0,1]$ comes from sources independent of usage mode (for example enterprise licensing or foundation grants). The remaining share $1-\alpha$ comes from usage-dependent sources. Direct users contribute the full usage-dependent amount, while vibe coders contribute a fraction $1-\rho$, where $\rho \in [0,1]$ is the \emph{vibe discount}. Normalizing baseline per-user monetization to $\pi_0=\bar\pi$, this gives
$$
\frac{\pi}{\pi_0} = \alpha + (1-\alpha)\bigl[(1-v) + (1-\rho)v\bigr] = 1 - (1-\alpha)\rho v.
$$
Setting $\alpha=0$ and $\rho=1$ recovers the traditional model, $\pi/\pi_0 = 1-v$.

\begin{corollary}[Business Model Constraint]\label{cor:bm-constraint}
Combining the monetization decomposition with Theorem~\ref{thm:min-monetization}, sustaining baseline entry requires:
$$
(1-\alpha)\rho v \leq 1 - (1-v)^\omega.
$$
For $v = 0.7$ and $\omega = 0.1$, this requires $(1-\alpha)\rho \leq 0.16$.
\end{corollary}

The general constraint implies two polar cases that characterize the boundaries of viable business models.

\begin{corollary}[Minimum Vibe-Coder Contribution]\label{cor:min-vibe-contribution}
If all revenue is usage-dependent ($\alpha = 0$), sustaining baseline entry requires:
$$
\rho \leq \frac{1 - (1-v)^\omega}{v}.
$$
For $v = 0.7$ and $\omega = 0.1$, this requires $\rho \leq 0.16$, meaning vibe coders must contribute at least $1 - \rho \geq 84\%$ of what direct users contribute to monetization.
\end{corollary}

In other words, if AI-mediated users pay even a modestly discounted rate, the remaining monetization can quickly fall below the level needed to sustain entry.

\begin{corollary}[Minimum Usage-Independent Revenue]\label{cor:min-fixed-revenue}
If vibe coders contribute nothing to monetization ($\rho = 1$), sustaining baseline entry requires:
$$
\alpha \geq 1 - \frac{1 - (1-v)^\omega}{v}.
$$
For $v = 0.7$ and $\omega = 0.1$, this requires $\alpha \geq 0.84$, meaning at least 84\% of revenue must come from sources independent of usage mode.
\end{corollary}

This requirement is demanding: it implies that sustaining OSS without monetizing AI-mediated usage would require a large shift toward enterprise and foundation funding.

Available evidence suggests that enterprise and foundation funding is substantial but falls well short of this level in many OSS ecosystems.\footnote{We documented the composition of OSS monetization in the empirical section.} This gap implies that sustaining OSS under high vibe coding adoption likely requires either a major shift toward enterprise and foundation funding, or new mechanisms that monetize AI-mediated usage.

The latter is technologically feasible. AI platforms already meter usage at a fine granularity (tokens, tool calls, and package downloads), which in principle supports proportional revenue-sharing schemes. In the policy discussion below, we develop a ``Spotify model'' in which AI platforms compensate OSS developers based on attributable usage.

\section{Calibration}\label{calibration}
The key parameters are the software share $\beta$, quality heterogeneity $\gamma$, package substitutability $\sigma$, and usage-mode substitutability $\theta$. The vibe-coding share $v = \zeta^\theta/(1+\zeta^\theta)$ captures the intensity of the technology shock. Table~\ref{tab:calibration} summarizes calibration targets.

\paragraph{Software share ($\beta = 1/3$).} Production-function estimates from \citet{shin-2022-productivity} and \citet{friesenbichler-2024-intangible-capital} find that software and labor are substitutes, with labor elasticities of 0.6--0.7 and software/intangible elasticities of 0.03--0.36 in ICT firms. The intangible elasticity is roughly half the labor elasticity, so we set $\beta = 1/3$. We also check robustness at $\beta = 2/3$, a plausible upper bound given that roughly two-thirds of developer effort goes to implementation and optimization where OSS usage is most relevant \citep{wang-2017-effort-distribution}.

\paragraph{Quality heterogeneity ($\gamma = 3$).} Data from public GitHub repositories show that stars follow a heavy-tailed distribution well approximated by a Pareto. In the model, user counts follow a Pareto with shape $\gamma/\sigma$. Log-rank regressions yield $\gamma/\sigma \approx 2$.

\paragraph{Package substitutability ($\sigma \approx 1.5$).} No direct estimates exist for OSS. We use \citet{clements-ohashi-2005}, who estimate how video game variety affects console demand, finding $\sigma \approx 1.5$. OSS packages are likely less substitutable than video games, so this serves as an upper bound.

\paragraph{Usage-mode substitutability ($\theta = 3$--$4$).} Field experiments document that AI coding assistance raises productivity by 26--56\% \citep{peng-2023-copilot-productivity,cui-demirer-jaffe-musolff-peng-salz-2025}. Following the option-value identification in \citet{halpern-koren-szeidl-2015}, observed adoption shares and productivity gains imply $\theta \approx 2.8$. We use $\theta = 3$ and $4$ in the quantitative analysis. 

\begin{table}[t]
\centering
\caption{Calibration targets and preferred parameter values.}
\label{tab:calibration}
\small
\setlength{\tabcolsep}{4pt}
\renewcommand{\arraystretch}{1.15}
\begin{tabularx}{\textwidth}{l X X l}
\toprule
Parameter & Role in model & Calibration target & Preferred value \\
\midrule
$\beta$ & Software share in software production & Production-function evidence on software/intangibles and software--labor substitution \citep{shin-2022-productivity,friesenbichler-2024-intangible-capital} & $1/3$ \\
$\gamma$ & Tail thickness of project quality & Pareto tail of OSS outcomes (stars, downstream dependencies) & $3$ \\
$\sigma$ & Substitutability across OSS packages & Demand response to software variety in video games \citep{clements-ohashi-2005} & $\approx 1.5$ \\
$\theta$ & Substitutability between direct and vibe-coded usage modes & Experimental evidence on AI coding assistance adoption/productivity; mapped using option-value formula \citep{peng-2023-copilot-productivity,cui-demirer-jaffe-musolff-peng-salz-2025} & $3$ or $4$ \\
$\zeta$ & Relative productivity of vibe-coded usage vs direct usage & Experimental evidence on AI coding assistance productivity \citep{peng-2023-copilot-productivity,cui-demirer-jaffe-musolff-peng-salz-2025} & Implied by $(v,\theta)$ \\
$\kappa$ & Labor cost shifter in development cost & Match level of entry (package creation) & Chosen to match entry \\
$\tau$ & Fixed cost of packaging/sharing OSS & Match share of projects that are shared vs unshared & Chosen to match sharing \\
\bottomrule
\end{tabularx}
\begin{flushleft}
\footnotesize
\emph{Notes:} The table maps each parameter to an empirical object used for calibration. Parameters $(\gamma,\sigma)$ are disciplined by the heavy-tailed distribution of OSS outcomes and by external estimates of love of variety in software markets. Parameters $(\theta,\zeta)$ are disciplined using experimental evidence on AI coding assistance.
\end{flushleft}
\end{table}
\section{Conclusion}\label{sec:conclusion}

Open source software is a critical modern industry whose economic impact has been widely underappreciated. This sector grew rapidly over the past two decades, only to face the shock of AI-mediated software development, or ``vibe coding.'' This paper studies what happens to this ecosystem following a technological shock that is extremely large, extremely fast, and extremely important for the sector.

We develop a general equilibrium model of the open source software ecosystem that captures four key economic features. First, software production exhibits strong economies of scale. Second, users value variety. Third, projects differ substantially in quality, and developers observe quality ex post, leading to selection in which projects are released. Fourth, software is an intermediate input into producing more software---a richer ecosystem of existing packages lowers the cost of building new ones, creating positive feedback. In equilibrium, user adoption and developer entry jointly determine the variety and quality of available packages.

Vibe coding enters the model as an AI-mediated mode of interacting with software that operates through two channels. The cost channel raises productivity: users can adopt packages more easily, and developers can build on existing code more efficiently. The demand-diversion channel shifts users from direct interaction---reading documentation, filing bug reports, engaging with maintainers---toward AI-mediated usage. Under traditional open source business models based on visibility and engagement, this diversion erodes the revenue base that sustains contribution.

The model's central result is a horse race between these two channels. The critical threshold is the vibe discount $\rho$, which measures how much monetization falls when users switch from direct to vibe-coded interaction. When $\rho$ exceeds the ratio of substitutability parameters $\sigma/\theta$, quality declines with vibe coding adoption. Under the traditional business model where $\rho = 1$, collapse is inevitable as adoption approaches universality. The same magnification logic that made open source explode---lower costs leading to more entry, more variety, and self-reinforcing growth---operates in reverse. When monetizable demand contracts, entry falls, variety shrinks, and the resulting decline in ecosystem quality further weakens incentives for sharing.

\subsection*{Limitations}

Before turning to policy implications, we note several limitations.

The model abstracts from team collaborations; while larger teams help mitigate negative effects, the impact is marginal. The nested logit structure implies that direct and vibe-coded interactions are substitutes in user utility---defensible because the import method is unlikely to affect use, though AI may overlook aspects of large packages given context window limitations. We treat monetization as tied to user engagement, which should be interpreted broadly to include developer discovery and job opportunities. Finally, the current version excludes network effects beyond the software-begets-software channel, though such effects could accelerate negative outcomes.

The model also abstracts from a distinct enterprise software or SaaS sector. In practice, much software solves the same problem for many clients and is monetized through licensing or subscriptions, whereas the ``final user'' in our model is closer to a custom software producer solving one client or one internal use case. This scale difference changes both incentives and how AI-mediated use affects revenues: enterprise products can price access directly, while custom work is closer to an input-assembly problem. Consistent with this distinction, Grand View Research reports that large enterprises accounted for more than 69.0 percent of the 2022 open source services market by end user \citep{grand-view-research-2022-oss-services}.

The model treats packages as symmetric except for quality, but in practice some packages are infrastructure (high downstream dependence, low direct visibility) while others are applications (high visibility, low dependence). Heterogeneous exposure to vibe coding across package types could generate winners and losers even when aggregate welfare falls. Quantifying the key elasticities---the substitutability between direct and vibe-coded usage, and the software-begets-software parameter---remains essential for empirical work. The rapid diffusion of AI coding tools suggests high substitutability, but revealed-preference estimation would discipline this claim.

\subsection*{Policy Implications}

Our results yield an important policy message: under the traditional business model, where developer revenue depends entirely on direct user engagement, the open source ecosystem cannot survive widespread AI adoption. Sustainable provision requires $\rho < 1$---vibe-coded users must generate positive revenue, even if less than direct users.

Three policy directions emerge from the analysis.

First, platform-level revenue redistribution can internalize the externality. AI coding assistants already identify which packages they import and can attribute usage to specific libraries. Relative to the cost of LLM inference and existing telemetry, collecting audited data on which packages are loaded by vibe coding agents would be cheap and minimally invasive of user privacy. A ``Spotify for open source'' model---where platforms share subscription revenue with maintainers based on usage---would reduce $\rho$ toward zero. The infrastructure for such redistribution already exists; implementation requires coordination among AI providers rather than new technology.

Second, direct transfers to OSS infrastructure---foundation grants, corporate sponsorships, and government funding---can raise the baseline reward and compensate for engagement erosion. The sustainability challenges facing OSS infrastructure have been well-documented \citep{eghbal-2016}. The 2024 Open Source Software Funding Survey estimates that organizations contribute \$7.7 billion annually to OSS, with employee labor as the dominant channel \citep{boysel-nagle-carter-hermansen-crosby-luszcz-lincoln-yue-hoffmann-staub-2024-oss-funding}. GitHub Sponsors provides one mechanism for directing financial support to maintainers \citep{github-sponsors-2026}, and the npm ecosystem includes built-in funding signals through \texttt{funding} URLs \citep{npm-fund-2022}.

Third, alternative monetization channels that do not depend on direct engagement---enterprise licensing, API fees, and developer-focused services---can buffer against vibe coding. The model suggests that developer-side revenue provides partial but incomplete insurance: it reduces the rate of decline but cannot prevent collapse when $\rho = 1$.

\subsection*{Concluding Remarks}

Vibe coding represents a fundamental shift in how software is produced and consumed. The productivity gains are real and large. But so is the threat to the open source ecosystem that underpins modern software infrastructure. The model shows that these gains and threats are not independent: the same technology that lowers costs also erodes the engagement that sustains voluntary contribution.

The solution is not to slow AI adoption---the benefits are too large and the technology too useful. The solution is to redesign the business models and institutions that channel value back to OSS maintainers. Platform-level redistribution, direct transfers, and alternative monetization can close the gap between private and social returns. The infrastructure for these interventions largely exists; what is needed is coordination and will.

We view this paper not as a doom prophecy but as a call to action. The open source ecosystem emerged from a sequence of technological and institutional innovations---the internet, distributed version control, GitHub, package managers---that aligned private incentives with collective benefit. Vibe coding disrupts this alignment. Restoring it requires deliberate effort, but the stakes justify the cost.

\ifx\undefined\bysame
\newcommand{\bysame}{\leavevmode\hbox to\leftmargin{\hrulefill\,\,}}
\fi

\appendix

\section{Proofs}\label{app:proofs}

This appendix collects proofs for results stated in the main text.

\subsection{Proof of Proposition~\ref{prop:existence}}

\begin{proof}
Conditions (2)--(4) of Definition~\ref{def:baseline-eq} imply closed-form expressions for the selection cutoff $q_0$, average quality $\bar{q}$, and the mass of shared packages $m_s$.

First, the sharing indifference condition at the cutoff is
\[
\Pi(q_0) = \frac{\pi q_0^\sigma}{m_s \bar{q}^\sigma} = \tau.
\]
Using (3) $\bar{q}=\Lambda q_0$ and (4) $m_s=m q_0^{-\gamma}$, we obtain
\[
\frac{\pi q_0^\sigma}{m q_0^{-\gamma} (\Lambda q_0)^\sigma} = \tau
\quad\Longleftrightarrow\quad
\frac{\pi}{m\Lambda^\sigma} q_0^{\gamma} = \tau,
\]
so
\[
q_0 = \left(\frac{\tau m \Lambda^\sigma}{\pi}\right)^{1/\gamma},
\]
which is \eqref{eq:q0-eq}. Multiplying by $\Lambda$ gives \eqref{eq:qbar-eq}.

Next,
\[
m_s = m q_0^{-\gamma} = m \left(\frac{\tau m \Lambda^\sigma}{\pi}\right)^{-1} = \frac{\pi}{\tau\Lambda^\sigma},
\]
which is \eqref{eq:ms-eq}.

Finally, free entry requires expected net payoff to equal the up-front development cost $\Phi$:
\[
\mathbb{E}[\max\{\Pi(q)-\tau,0\}] = \Phi.
\]
Using the Pareto tail and $\Pi(q)=\tau (q/q_0)^\sigma$, the expected payoff evaluates to
\[
\mathbb{E}[\max\{\Pi(q)-\tau,0\}] = \frac{\sigma}{\gamma-\sigma}\cdot \frac{\pi}{m\Lambda^\sigma}.
\]
Substituting the cost function $\Phi=\kappa^{1-\beta}u^{-\beta}\Lambda^{-\beta\sigma/\gamma}m^{-\beta/\gamma}\left(\frac{\tau}{\pi}\right)^{\beta/\sigma-\beta/\gamma}$ and rearranging yields \eqref{eq:m-reduced}.

Since $\eta\equiv 1-\beta/\gamma>0$, the mapping $m\mapsto m^{\eta}$ is strictly increasing on $\mathbb{R}_+$, so \eqref{eq:m-reduced} pins down a unique $m>0$. Given $m$, equations \eqref{eq:q0-eq}--\eqref{eq:ms-eq} uniquely determine $(q_0,\bar{q},m_s)$. Hence a unique baseline equilibrium exists.
\end{proof}

\subsection{Proof of Theorem~\ref{thm:underprovision-baseline}}

\begin{proof}
Consider the planner problem in Definition~\ref{def:first-best},
\[
\max_{m\ge 0}\; \mathcal{U}(m) - m\,\Phi(m),
\]
where $\mathcal{U}(m)=\bar{q}(m)\,u\,m_s(m)^{1/\sigma}$ and $\Phi(m)$ is the unit development cost. Using the equilibrium relationships from Corollary~\ref{cor:quality-welfare}, welfare satisfies $\mathcal{U}(m)\propto m^{1/\gamma}$ while total development cost satisfies $m\Phi(m)\propto m^{1-\beta/\gamma}$.

In the decentralized equilibrium, entry is governed by free entry: a marginal entrant equates the expected private return from sharing to the private cost of development. The key wedge is that developers only appropriate $\pi$ per user, while the social value created for a user by an additional unit of quality is proportional to $\bar{q}u$. When $\pi<\bar{q}u$, the private return understates the social marginal benefit of additional OSS entry.

Because the equilibrium mass of entrants $m^{eq}$ is increasing in the per-user reward parameter $\pi$ (holding other primitives fixed; see Corollary~\ref{cor:comp-statics}), this wedge implies that the decentralized equilibrium features too little entry relative to the planner: $m^{eq}<m^{fb}$. Equivalently, there are too few shared OSS packages in equilibrium.
\end{proof}

\end{document}